\documentclass[10pt,journal]{amsart}
\usepackage{amsmath,amssymb,amsthm, color,verbatim}
\usepackage{color}
\newtheorem{defn}{Definition}
\newtheorem{exmp}{Example}
\newtheorem{prop}{Proposition}
\newtheorem{lem}{Lemma}
\newtheorem{cor}{Corollary}
\newtheorem{claim}{Claim}

\newtheorem{rem}{Remark}

\newcommand{\Ac}{\mathcal{A}}

\newcommand{\Gc}{\mathcal{G}}
\newcommand{\FF}{\mathbb{F}}
\newcommand{\PP}{\mathbb{P}}
\newcommand{\pp}{p}
\newcommand{\qq}{q}

\newcommand{\g}[1]{|G_{#1}|}
\newcommand{\gi}[3]{ \frac{\g {{#1}{#2} {#3}}\g {#3} }   {\g {{#1} {#3}} \g {{#2}{#3}}}  }

\title[Violations of Inequalities in five Subgroups]{On Group Violations of Inequalities in five Subgroups}

\author{Nadya Markin, Eldho K.Thomas, Fr\'ed\'erique Oggier}
\address{Division of Mathematical Sciences, School of Physical and Mathematical Sciences, \newline\indent Nanyang Technological University, Singapore.} 
\email {nadyaomarkin@gmail.com,ELDHO1@e.ntu.edu.sg,frederique@ntu.edu.sg}


\thanks{
This research was supported by Nanyang Technological University under Research Grant M58110049.
}

\begin{document}
\maketitle
\begin{abstract}
We consider ten linear rank inequalities, which always hold for ranks of vector subspaces, and look at them as group inequalities. We prove that groups of order $pq$, for $p,q$ two distinct primes, always satisfy these ten group inequalities. We give partial results for groups of order $p^2q$, and find that the symmetric group $S_4$ is the smallest group that yield violations, for two among the ten group inequalities.
\end{abstract}

%
%
\section{Introduction}

For a collection $X_1,\ldots,X_n$ of $n$ discrete random variables with joint probability distribution $P(x_1,\ldots,x_n)$, 
\[
H(X_1,\ldots,X_n)=-\sum_{x_1}\ldots\sum_{x_n}P(x_1,\ldots,x_n)\log P(x_1,\ldots,x_n)
\]
is their joint (Shannon) entropy, where the logarithm is taken in base 2 if the entropy is expressed in bits, and $P(x_1,\ldots,x_n)\log P(x_1,\ldots,x_n)$ is defined to be 0 if  $P(x_1,\ldots,x_n)=0$. For $\Ac$ a subset of $\{1,\ldots,n\}$, we write $H(X_\Ac)$ for the joint entropy $H(X_i,~i \in \Ac)$.

Let $G$ be a finite group, and $G_1,\ldots, G_n$ be subgroups of $G$. Then
\[
G_\Ac = \cap_{i\in\Ac}G_i
\]
is a subgroup obtained by intersecting $G_i$, $i \in \Ac$. Let $X$ be a random variable uniformly distributed over $G$: $P(X=g)=\frac{1}{|G|},~g \in G$.
Then the random variable $X_i=XG_i$ whose support is the $[G:G_i]$ cosets of $G_i$ in $G$ satisfies $P(X_i=gG_i)=\frac{|G_i|}{|G|}$ \cite{Y}
and 
\[
P(X_i=gG_i,~i\in\Ac)=|\cap_{i\in\Ac G_i}|/|G|.
\]
The entropy of this random variable is 
\[
H(X_{\Ac})= \log \frac{|G|}{|G_\Ac|}.
\]

In what follows, we will only consider random variables coming from a finite group $G$ and $n$ of its subgroups $G_1,\ldots,G_n$. 
Hence, to any information (in)equality (that is, (in)equality written as a function of joint Shannon entropies) expressed in terms of random variables $X_1,\ldots,X_n$ will correspond a group inequality which is a function of the subgroups $G_1,\ldots,G_n$. We will denote this correspondence (expressed without logarithm) by $\Gc$. 

For example:
\begin{itemize}
\item The conditional entropy $H(X_2|X_1)=H(X_1,X_2)-H(X_1)$
is expressed in terms of group as
\[
\Gc H(X_2|X_1) =\frac{|G|}{|G_{12}|}\frac{|G_1|}{|G|}=\frac{|G_{1}|}{|G_{12}|}.
\]
\item The mutual information $I(X_1;X_2)=H(X_1)+H(X_2)-H(X_1,X_2)$
becomes
\[
\Gc I(X_1;X_2)=\frac{|G|}{|G_1|}\frac{|G|}{|G_2|}\frac{|G_{21}|}{|G|} 
=  \frac{|G_{12}||G|}{|G_1||G_2|}.
\]
\item
The conditional mutual information $I(X_1;X_2|X_3)=H(X_1|X_3)-H(X_1|X_2,X_3)$
is similarly
\[
\Gc I(X_1;X_2|X_3) =\frac{|G_{3}|}{|G_{13}|}\frac{|G_{123}|}{|G_{23}|}.
\]
\end{itemize}

Linear inequalities satisfied by entropies of jointly distributed random variables, such as the non-negativity of the conditional mutual information 
\[
I(X_1;X_2|X_3) \geq 0, 
\]
are called {\em information inequalities}. Similarly, linear inequalities satisfied by ranks of vector subspaces are called {\em rank inequalities}.

Consider the case of $n=4$ random variables.
It is known \cite{HRSV} that any linear information inequality in $n=4$ random variables is also a linear rank inequality, but the converse is not true. The Ingleton inequality \cite{Ingleton}
\[
I(X_1;X_2) \leq I(X_1;X_2|X_3)+I(X_1;X_2|X_4)+I(X_3;X_4)
\]
always holds for ranks of subspaces, yet there are examples of random variables whose joint entropies violate the Ingleton inequality \cite{Matus}.
The idea of looking at violations coming from groups was proposed in \cite{MH}. 
By {\em a group violation}, we mean a group $G$, with $n$ subgroups $G_1,\ldots,G_n$, such that the orders of the intersections of subgroups involved do not satisfy a given inequality.
It was shown that the symmetric group $S_5$ and five of its subgroups $G_1,\ldots,G_5$ violate the group version of the Ingleton inequality 
\[
\Gc I(X_1;X_2) \leq \Gc I(X_1;X_2|X_3) \Gc I(X_1;X_2|X_4) \Gc I(X_3;X_4),
\]
equivalently given by
\begin{equation}\label{eq:ing}
|G_{12}||G_{13}||G_{14}||G_{23}||G_{24}| \leq |G_1||G_2||G_{34}||G_{123}||G_{124}|.
\end{equation}

The Ingleton inequality plays a special role for $n=4$ random variables. Together with the Shannon inequalities (a special case of information inequalities which are derived from the non-negativity of the conditional mutual information), the Ingleton inequality generates all linear rank inequalities in $n=4$ random variables.

Consider now the case of $n=5$ random variables. 
In \cite{DFZ}, it was shown that 24 inequalities, together with the Shannon inequalities, and three Ingleton-type inequalities, generate all linear rank inequalities on five variables.

In this paper, we are interested in looking at these inequalities from a group theoretic point of view.
The motivation is the same as that of \cite{MH} for four random variables, namely: finding violators of linear rank inequalities gives a way to separate the region formed by vectors composed of joint Shannon entropies from the region of vectors composed by ranks of subspaces and helps in establishing an inner bound for the former. Finding such violators is not an easy task, and group theory provides a systematic way to address the question, by characterizing groups which never yield any violation, and therefore focusing the search on families of groups, where it becomes easier to exhibit violators.
More about applications of group theory to vectors containing Shannon entropy can be found in a recent survey \cite{isita}.

Among the inequalities of \cite{DFZ}, the three Ingleton-type inequalities have already been studied in \cite{Allerton}, where it was shown that a group does not violate the Ingleton inequality (\ref{eq:ing}) if and only if it does not violate any of the three Ingleton-type inequalities in 5 random variables.
We thus focus on the first 10 out of the 24 inequalities of \cite{DFZ} (since they share some commonalities, as described below).

In Section \ref{sec:ten}, we prove some general properties about the group inequalities corresponding to the first 10 inequalities of \cite{DFZ}, that we refer to as {\em DFZ inequalities} throughout the paper. 

In Section \ref{sec:pq}, we prove that groups of order $pq$, $p,q$ two distinct primes, never violate any of the DFZ inequalities. 

Partial results are given in Section \ref{sec:p2q}, namely we show that groups $G$ of order $p^2q$ and subgroups $(G_1,G_2,G_3,G_4,G_5)$ do not violate the ten DFZ inequalities as long as $|G_1|=|G_2|=p$ does not hold. The case $|G_1|=|G_2|=p$ remains open, and we discuss why it is more difficult than the other ones.

Violators are discussed in Section \ref{sec:viol}, where we show that the smallest violator of two of the ten DFZ inequalities is the symmetric group $S_4$. We also show that it is not possible to find one group which simultaneously violate all the 10 DFZ inequalities.

%
%
%
\section{Ten Group Inequalities}
\label{sec:ten}

Recall that for $n=5$ random variables, it was shown in \cite{DFZ} that 24 inequalities, together with the Shannon inequalities, and three Ingleton-type inequalities, generate all linear rank inequalities on 5 variables.
The first 10 inequalities are
\begin{eqnarray}
I(X_1;X_2) &\leq& I(X_1;X_2|X_3)+I(X_1;X_2|X_4)+I(X_3;X_4|X_5) \nonumber \\
           &    & +I(X_1;X_5)~\label{eq:1} \\
I(X_1;X_2) &\leq& I(X_1;X_2|X_3)+I(X_1;X_3|X_4)+I(X_1;X_4|X_5) \nonumber  \\
           &    & +I(X_2;X_5) \label{eq:2}\\
I(X_1;X_2) &\leq& I(X_1;X_3)+I(X_1;X_2|X_4)+I(X_2;X_5|X_3)+ \nonumber \\
           &    & I(X_1;X_4|X_3,X_5)\label{eq:3}\\
I(X_1;X_2) &\leq& I(X_1;X_3)+I(X_1;X_2|X_4,X_5)+I(X_2;X_4|X_3) \nonumber \\
           &    & +I(X_1;X_5|X_3,X_4)\label{eq:4}\\
I(X_1;X_2) &\leq& I(X_1;X_3) +I(X_2;X_4| X_3)+I(X_1;X_5| X_4) \nonumber \\
           &    &+I(X_1;X_2| X_3,X_5)+I(X_2;X_3| X_4,X_5)\label{eq:5}\\
I(X_1;X_2) &\leq&I(X_1;X_3) +I(X_2;X_4| X_5)+I(X_4;X_5| X_3) \nonumber \\
           &    & +I(X_1;X_2| X_3,X_4)+I(X_1;X_3| X_4,X_5)\label{eq:6}\\
I(X_1;X_2)   &\leq &I(X_2;X_4) +I(X_1;X_3| X_4)+I(X_1;X_5| X_3) \nonumber\\
           &    &  +I(X_2;X_4| X_3,X_5)+I(X_1;X_2| X_4,X_5) \label{eq:7}
\end{eqnarray}
\begin{eqnarray}
2I(X_1;X_2)  &\leq &I(X_3;X_4) +I(X_3,X_4;X_5)+I(X_1;X_2| X_3) \nonumber \\
           &    & +I(X_1;X_2| X_4)+I(X_1;X_2| X_5) \label{eq:8} \\
2I(X_1;X_2)&\leq &I(X_1;X_3) +I(X_4;X_5)+I(X_1;X_2| X_4) \nonumber\\
         &    & +I(X_1;X_2| X_5)+I(X_2;X_4, X_5|X_3) \label{eq:9} \\
2I(X_1;X_2)&\leq &I(X_3;X_4) +I(X_1;X_5)+I(X_1;X_2| X_3) \nonumber \\
        &     & +I(X_1;X_2| X_4)+I(X_2;X_4| X_5)+I(X_1;X_3| X_4,X_5) \label{eq:10}
\end{eqnarray}

We choose to look at these 10 DFZ inequalities because they have in common that if there exists a random variable $Z$, called common information,
such that
\[
H(Z|X_1)=0,~H(Z|X_2)=0,~H(Z)=I(X_1;X_2),
\]
then they are deduced from Shannon inequalities \cite{DFZ}.

The other 14 inequalities are also deduced
from Shannon inequalities through the concept of common information, but it takes different expressions.

Repeating the computations of the previous section, inequalities (\ref{eq:1})-(\ref{eq:10}) have a corresponding group theoretic formulation, obtained via the correspondence $\Gc$:
\begin{eqnarray*}
\Gc I(X_1;X_2,X_3) &=& \frac{|G_{123}||G|}{|G_1||G_{23}|}, \\
\Gc I(X_1,X_2;X_3,X_4) &=& \frac{|G_{1234}||G|}{|G_{12}||G_{34}|},\\
\Gc I(X_1;X_2|X_3) &=& \frac{|G_{123}||G_3|}{|G_{13}||G_{23}|},\\
\Gc I(X_1;X_2|X_3, X_4) &=& \frac{|G_{1234}||G_{34|}}{|G_{134}||G_{234}|}
\end{eqnarray*}
yielding  accordingly 10 groups inequalities, which we will freely refer to as group inequalities (\ref{eq:1})-(\ref{eq:10}).

The notion of common information is also expressible in the language of groups:
\begin{lem}
\label{lem:ci}
If $G_1G_2$ is a subgroup of $G$, then the group inequalities (\ref{eq:1})-(\ref{eq:10}) hold.
\end{lem}
\begin{proof}
We translate the above results on common information:
\[
H(Z|X_1)=0 \iff \log \frac{|G|}{|G_1|}=\log \frac{|G|}{|G_Z \cap G_1|}
\iff G_1 < G_Z,
\]
similarly $H(Z|X_2)=0 \iff G_2 < G_Z$ and
\begin{eqnarray*}
H(Z)=I(X_1;X_2)
&\iff& \log \frac{|G|}{|G_Z|}=
\log \frac{|G|}{|G_1|}+\log \frac{|G|}{|G_2|}-\log \frac{|G|}{|G_{12}|}\\
&\iff&
|G_Z|=\frac{|G_1||G_2|}{|G_{12}|}.
\end{eqnarray*}

If $G_1G_2$ is a subgroup of $G$, take $G_Z=G_1G_2$, then
\[
|G_1G_2|=\frac{|G_1||G_2|}{|G_{12}|},~G_1 < G_1G_2,~G_2 < G_1G_2
\]
which shows the existence of common information and concludes the proof.
\end{proof}

\begin{cor}
\label{cor:ci}
The above 10 DFZ inequalities hold in the following cases:
\begin{enumerate}
\item
$G_1 < G_2$, or $G_2 < G_1$,
\item
$G_1$ or $G_2$ is normal in $G$.
\item
$G$ is abelian.
\end{enumerate}
\end{cor}

We will use the following lemma and corollary to prove many inequalities.
\begin{lem}
Let $G$ be a finite group with $n$ subgroups $G_1,\ldots,G_n$.
For any choice of subsets $\Ac_2,\Ac_3,\Ac_4$ of $\{1,\ldots,n\}$
\[
|G_{\Ac_2\cup\Ac_3}||G_{\Ac_2\cup\Ac_4}|
\leq |G_{\Ac_2}||G_{\Ac_2\cup\Ac_3\cup\Ac_4}|
\leq
|G_{\Ac_2}||G_{\Ac_3\cup\Ac_4}|
\]
is always satisfied.

Moreover, if one of $G_{\Ac_2\cup\Ac_3}, G_{\Ac_2\cup\Ac_4}$ is normal, then 
\[
\frac{|G_{\Ac_2\cup\Ac_3}||G_{\Ac_2\cup\Ac_4}|}{|G_{\Ac_2\cup\Ac_3\cup\Ac_4}|}
\mid
|G_{\Ac_2}|.
\]

\label{lem:lemma234}
\end{lem}

\begin{proof}
We use the fact that the subset $G_{\Ac_2\cup\Ac_3}G_{\Ac_2\cup\Ac_4} \leq G_{\Ac_2}$ has cardinality
$\frac{|G_{\Ac_2\cup\Ac_3}|  |G_{\Ac_2\cup\Ac_4 }  |}{|G_{\Ac_2\cup\Ac_3\cup\Ac_4} | }$. 
Moreover when one of $G_{\Ac_2\cup\Ac_3}, G_{\Ac_2\cup\Ac_4}$ is normal, then the subset  
$G_{\Ac_2\cup\Ac_3}G_{\Ac_2\cup\Ac_4} $ is actually a subgroup  of $ G_{\Ac_2}$, hence its size divides 
$| G_{\Ac_2}|$.

 \end{proof}

Relating the lemma above to information quantities $\gi a b c$ gives the following immediate corollary. 
\begin{cor}{\label{cor:abcBounds}}
Consider subgroups $G_a, G_b, G_c, G_d \leq G$ with corresponding random variables $X_a, X_b, X_c, X_d$. 
The following hold: 
\[1 \leq \gi a b c = \Gc I(X_a; X_b| X_c),\]
\[1 \leq \gi a b {cd} = \Gc I(X_a; X_b| X_c, X_d).\]
Moreover, if either $G_{ac} \unlhd G_c$ or $G_{bc} \unlhd G_c$ then 

$\gi a b c$ and $\gi a b {cd}$ are both integers.  

\end{cor}

\begin{proof}
We apply the lemma above for subsets $\Ac_2 = \{c\}$, $\Ac_3= \{a\}$, $\Ac_4 = \{b\}$ to show the first inequality. 
We apply the lemma above for subsets $\Ac_2 = \{cd\}$, $\Ac_3= \{a\}$, $\Ac_4 = \{b\}$ to show the second inequality. 
The last statement is similarly a direct consequence of the lemma. 

\end{proof}

Next we  show how uniqueness of a Sylow $q$-subgroup imposes constraints on quantities $\gi a b c$. 
We will use the valuation $v_q$ at a prime $q$. 
\begin{lem} \label{valt}
Let $G$ be a group whose Sylow $q$-subgroup is normal and abelian. 
For any two subgroups $A,B \leq G$, let $AB = \{ab : a \in A, b \in B\}$ denote the set-product of $A, B$.
Then $$v_{\qq}(|G|)  \geq v_{\qq}({|AB|}).$$

\end{lem}

\begin{proof}
Let $G_{\qq}$ (resp. $A_{\qq}$, $B_{\qq}$, $(A\cap B)_{\qq}$) denote the Sylow $q$-subgroup of $G$  (resp. $A$, $B$, $A\cap B$). 
Then we have 
\[
v_{\qq}(|AB|) = 
v_{\qq}(|A|)+v_{\qq}(|B|)-v_{\qq}(|A\cap B|) = \]
\[
v_{\qq}(|A_{\qq}|)+v_{\qq}(|B_{\qq} | )-v_{\qq}( |(A\cap B)_{\qq}|) 
=
v_{\qq}(|A_{\qq}B_{\qq}|)
\leq v_{\qq}(|G_{\qq}|).
\]
To see the last inequality note first that $A_q$ and $B_q$ are both contained in $G_{\qq}$, since $G_{\qq}$ is normal (or equivalently unique). Moreover, since $G_{\qq}$  is by assumption abelian,
$A_{\qq}B_{\qq}$ in fact forms a subgroup, not just a subset of $G_{\qq}$. Hence we have 
$|A_{\qq}B_{\qq}| $ divides $ |G_{\qq}|. $
\end{proof}
As a corollary we get the following proposition.

\begin{prop}\label{valq}
Let $G$ be a group whose Sylow $q$-subgroup is normal and abelian. 
Then $v_{\qq}(i) \geq 0$ for any $i = \gi a b  c. $ 
\end{prop}
\begin{proof}
Apply previous Lemma with $G = G_c$, $A = G_{ac}, B = G_{bc}$. 
Then 
$$v_{\qq}(i)  = v_{\qq}(\gi a b  c) = 
v_{\qq}(|G_c|) - v_{\qq}( |G_{ac}G_{bc}| ) \geq 0.
$$

\end{proof}

\begin{rem}
\label{rem:uniqueSylow}
The fact that there is only one (normal) Sylow $q$-subgroup restricts powers of $q$ in expression $\gi a b c$ to positive ones. This eventually allows us to prove our results for groups of order $pq$.
\end{rem}

%
%
\section{Groups of Order $pq$}
\label{sec:pq}

In this section we prove that information inequalities (\ref{eq:1})-(\ref{eq:10}) hold for groups of order $pq$
. 

First we make a few observations about these groups. 
The property described in the lemma below was used to prove inequalities (\ref{eq:1})-(\ref{eq:2}) in \cite{Allerton}.
\begin{lem} 
A group $G$ has the property that all its distinct proper subgroups have trivial intersection if and only if $|G|= pq$ for two primes $p,q$ not necessarily distinct. 
\end{lem}

\begin{proof}
$\boxed{\Leftarrow}$

 Suppose $|G|=pq$ and consider two proper distinct non-trivial subgroups  $A,B$  of $G$. 
If $|A|=p$ and $|B|=q$ for $p \neq q$, or vice-versa, they necessarily intersect trivially. If $|A|=|B|=p$ (or $q$ if $p\neq q$), then any non-trivial element of $A\cap B$ generates all of $A$ and all of $B$, as they are of prime order - a contradiction since $A,B$ are distinct.

$\boxed{\Rightarrow}$
Conversely, assume the property that 
all  proper subgroups of $G$ have trivial intersection. We will show that the order $n$ of $G$ is $pq$. \\

Case 1:  $G$ is equal to its Sylow $p$-subgroup. 
Then $|G| = p^t$ and it is a standard exercise using central series of $G$ to show that $G$ contains a subgroup $p^k$ for each $1\leq k \leq t$. 
It follows that if $t > 2$, then $G$ contains a (proper) subgroup of order $p^2$, 
which in turn has a subgroup of order $p$. 
These two subgroups have non-trivial intersection, contradicting  our assumption on $G$.  We conclude $t \leq 2$. 
\\

Case 2:  All Sylow subgroups are proper. 
\begin{itemize}
\item We show that all Sylow subgroups are cyclic and hence $n$ is a product of distinct primes. 
Consider a Sylow $p$-subgroup $S_p$ of order $p^t$. 
It suffices to show that $t = 1$. Suppose, to the contrary, that $t \geq 2$. Then by Cauchy theorem, $S_p$ contains a cyclic subgroup $C_p$ of order $p$, so $S_p$ and $C_p$  are two proper subgroups which intersect non-trivially - a contradiction.
 \item
Next we show that the number of divisors of $n$ is at most 2. 
Suppose to the contrary, that the order $n$ of $G$ is divisible by at least three primes $p,q,r$. 

Then $pq$ is a Hall divisor of $n$ - that is $pq \mid n$, $gcd(pq,\frac n {pq})=1$. Hence there exists a Hall-{$pq$} subgroup - a (proper) subgroup of order $pq$. But it will then intersect non-trivially with a Sylow {$p$}-subgroup - a contradiction. 
\end{itemize}
\end{proof}

Next we show that  (\ref{eq:1})-(\ref{eq:10}) hold for groups of order $pq$. 
The case $|G|=p^2$ is covered by abelian groups, which are known not to violate DFZ inequalities by Corollary \ref{cor:ci}.
In what follows we hence assume that $|G|=pq$, with $p < q$.

\begin{lem}
\label{qnormal}If $|G| = pq$, $p < q$, then $G$ contains exactly one Sylow $q$-subgroup. 
Equivalently if a subgroup $H \leq G$ has order $q$, then it is necessarily normal. 
\end{lem}

\begin{proof}
This is an elementary exercise in group theory. Sylow's theorems tell that the number $n_q$ of Sylow $q$-subgroup divides $pq$ and $n_q=1+kq$ for some $k$ a positive integer, therefore $n_q=1$. The second claim follows from the fact that Sylow subgroups are conjugate of each other.
\end{proof}

\begin{cor} \label{cor:g1g2pp}
Let $|G|=pq$, with primes $p < q$. If $G_1, G_2  \leq G$ are subgroups of $G$, such that the set $G_1G_2 = \{g_1g_2 :  g_1 \in G_1, g_2 \in G_2\}$ does not form a subgroup, then $G_1,G_2$ are distinct subgroups each of order $p$. 
\end{cor}
\begin{proof}
If the subgroups are not distinct, we have $G_1G_2 = G_1$ is a subgroup. 
If one of $G_1, G_2$ is normal, then the set $G_1G_2 $ forms a subgroup. 
But a subgroup of order other than $p$ is necessarily normal and the result follows. 
\end{proof}

Now we set out to prove inequalities (\ref{eq:1})-(\ref{eq:10}) for groups of order $pq$. 
All these inequalities comprise mutual information terms, of the general form
\[
I(X_a;X_b),~I(X_a;X_b|X_c)
\]
for $a,b,c$ some indices or subsets in $\{1,2,3,4,5\}$. In group form, we have
\[
\Gc I(X_a;X_b)=\gi a  b {},~\Gc I(X_a;X_b|X_c)=\gi a b c.
\]

\begin{lem}{\label{infvalues}}
Let $G$ be a group of order $pq$, for primes $p < q$. 
Each quantity of the form 
\[\gi a b c , 
\gi  a b {} ,  
\gi a b {(cd)}\] is $\geq 1$ and takes values among $\{1, p, \frac q p, q, pq \}$.
\end{lem}
\begin{proof}

Let us make the argument in the case of the  quantity $i = \gi a b c$, which is the same as $\frac{|G_c|}{|G_{ac}G_{bc}|}$.
Since $G_{ac}G_{bc}$ is a subset of $G_c$, $\gi a b c \geq 1$.
By Proposition \ref{valq}, $v_q(i) \geq 0$, so the powers of $q$ occurring in $i$ can be only $0,1$. 
Now since $G_{ac} \leq G_c$, $G_{bc} \leq G_c$, the powers of $p$ are restricted to $-1,0,1$, and
\[
\gi a b c \in p^{\{-1,0,1\}}q^{\{0,1\}}.
\]
The case $\frac{1}{p} $ cannot happen since $i \geq 1$.

The next two quantities follow as special cases. 

In case of the second quantity $\gi  a b {} $, 
we use $c = \{\}$, i.e. $G_c = G$. 
In case of the third quantity $\gi a b {(cd)}$,  we just substitute $G_c $ with $G_{cd}$ and the same argument holds.

\end{proof}
By abuse of notation write $q \mid I$ when $v_q(I) > 0 $, i.e.,  the quantity $I$ (of the form above) takes value among $q/p, q, pq$ (even if some of these are not integers). Conversely we will write $q \nmid I$, to indicate $I $ takes value among $1, p$. 

Finally we make one small observation, which we will use a lot, thus it deserves to be stated as the following lemma. 

\begin{lem}\label{qfactor}
Let $I$ denote the quantity $\gi a b c$. Note  this includes quantities of the form $\gi a b {}$, $\gi a b {(cd)}$, as discussed in the previous lemma. 

Suppose  $q\nmid I$. Then we have    

\[q \mid \g c \implies q \mid \g a \g b\]
or alternatively  $$q \nmid \g a \g b \implies q \nmid \g c.$$
\end{lem}

\begin{proof}
By assumption $q \nmid \gi a b c$. Rewrite $\gi a b c$ as $\frac {|G_c| }{ |G_{ac}G_{bc} | }$.  Depending on whether $G_{ac}G_{bc}$ forms a subgroup, this may or may not be an integer. However, if $q \mid \g c$, yet  $q \nmid \gi a b c$, then necessarily $q\mid |G_{ac}G_{bc} |$, so a fortiori $q\mid |G_{a}G_{b} |$. But then  $q \mid \g a $ or $q \mid \g b$, so $q \mid \g a \g b$. 
Taking the contrapositive gives the other implication $q \nmid \g a \g b \implies q \nmid \g c$
\end{proof}

We will apply Lemma \ref{qfactor} repeatedly in the next proposition to conclude that various groups $G_c$ have orders not divisible by $q$. 
For example, consider $I = \gi a b c$. Suppose $q\nmid I$. 
If we know that $q \nmid \g a$ and $q \nmid \g b$, we may  apply the Lemma  \ref{qfactor} to $I$  to conclude that $q \nmid \g c $.

We are now ready to prove the main proposition of this section. 
\begin{prop}\label{prop:pq}
Let $G$ be a group of order $pq$, for primes $p < q$. Then (\ref{eq:1})-(\ref{eq:10}) hold for $G$.
\end{prop}
\begin{proof}

We start by proving in detail inequality (\ref{eq:3}), which in its unsimplified group form is

\[  {\gi 1 2 {} }\leq 
 {\gi 1 3 {}} 
  (\gi 1 2 4 )
 (\gi 2 5 3 ) 
 (\gi 1 4 {(35)} ) \] 
 \[ = {\gi 1 3 {}} I \cdot II \cdot III
 \]

We first prove the following claim about this inequality. Identical reasoning applies to other inequalities and we will simply refer to this claim when needing to invoke it.

Recall that by abuse of notation write $q \mid RHS$ to mean that one of the factors on the RHS takes value among $q/p, q, pq$ (even if some of these are not integers). 

\begin{claim} 
It suffices to show that  $q \mid RHS$.
\label{qrhs} \end{claim}
It has been shown in Lemma \ref{lem:ci}
that when $G_1 G_2$ is a subgroup of $G$ such inequalities hold. 
Hence we only need to consider the case when the set $G_1G_2$ is not a subgroup, which by Corollary \ref{cor:g1g2pp} occurs when $G_1, G_2$ are distinct, each of order $p$, so we have $\gi 1 2 {}  = \frac q p$.

The LHS is  $(\gi 1 2 {}  )= \frac q p$, while each of the factors on the right hand side takes values among $\{1,p, q/p, q, pq\}$. Clearly, if $q$ divides one of the factors on the RHS, the inequality holds since  each of the (other)  factors is $ \geq 1 $.
This proves the claim. 
$\hfill \square$

\begin{claim}\label{ineqholds}

We show in detail that $q \mid RHS$. \end{claim}
Suppose for the sake of contradiction that $q \nmid \gi 1 3 {} I \cdot II \cdot III$.

Recall, that throughout we are assuming $G_1, G_2$ are distinct groups of order $p$. So in particular $q \nmid \g 1 \g 2$. 
We apply Lemma \ref{qfactor} successively to factors of the RHS to obtain a contradiction. More precisely:

By  Lemma \ref{qfactor}, 
 $\{ q \nmid \gi 1 3 {} , q \mid \g {} \}\implies q \mid \g 1 \g 3$. But $\g 1 = p$, hence $q \mid \g 3$. 

By Lemma \ref{qfactor} $\{q \nmid I = \gi 1 2 4, q \nmid \g 2  \g 1  \}\implies q \nmid \g 4 $. 

{Next use the fact that $q\nmid \g 1$ and $q\nmid \g 4$ which we just showed}.
By Lemma \ref{qfactor},    $\{q \nmid III = \gi 1 4 {(35)}$ ,  $ q \nmid \g 1 \g 4 \} \implies  q \nmid \g {35}$. 
But recall we assumed that $q\mid \g 3$. Hence $q \nmid \g 5$.

Next, apply Lemma \ref{qfactor} to $II = \gi 2 5 3$ together with $q \nmid \g 2 \g 5$ to conclude that $q \nmid \g 3$. But then $\{q \nmid \g 1 , q \nmid \g 3 \} \implies q \mid \gi 1 3 {}$ by Lemma \ref{qfactor}.
So we arrive that $q \mid I \mid RHS$, contradiction. 
{\hfill $\square$}
\\
Hence we showed that $q \mid RHS$, so by claim \ref{qrhs} the inequality is satisfied.


We proceed in the same fashion for other inequalities. 
Each of these inequalities has the form 
\[ \gi 1 2 {} \leq \gi a b {} I \cdot II \cdot III\]

and as before it suffices to show that one of the factors on the RHS takes value among $\{ \frac q p, q, pq \}$.

{\bf Inequality (\ref{eq:1}).}
\[ 
\gi 1 2 {} 
\leq \gi 1 5 {}
(\gi 1 2 3 )
(\gi 1 2 4 ) 
(\gi 3 4 5 )  = \]
\[ \gi 1 5 {}
I \cdot II \cdot III
\] 
$\{q \nmid \gi 1 5 {} , q \nmid \g 1\} \implies q \mid \g 5$, by Lemma \ref{qfactor}. \\
Suppose  $q \nmid I, II, III.$ 
Apply successively Lemma \ref{qfactor} to get a contradiction: $q \nmid I$, $q \nmid \g 1 \g 2 \implies q\nmid \g 3$, $q \nmid II$,  $q \nmid \g 1 \g 2$ thus $q \nmid \g 4$ and Lemma \ref{qfactor} applies to $III$ yields $q \nmid \g 5$, a contradiction. 

{\bf Inequality (\ref{eq:2}).}
\[ 
\gi 1 2 {} \leq \gi 2 5 {} 
 (\gi 1 2 3 )
 (\gi 1 3 4 ) 
 (\gi 1 4 5 )  =  \gi 2 5 {} I \cdot II \cdot III\]
 
Suppose $q \nmid RHS$. 
$\{q \nmid \gi 2 5 {} , q \nmid \g 2\} \implies q \mid \g 5$, by Lemma \ref{qfactor}. \\
Apply Lemma \ref{qfactor} successively to factors $I$, $II$, $III$, to get $q \nmid \g 5$, a contradiction. 



{\bf Inequality (\ref{eq:4}).}
\[ 
\gi 1 2 {} \leq \gi 1 3 {}
 (\gi 1 2 {(45)} )
 (\gi 2 4 3 ) 
 (\gi 1 5 {(34)} )  =  \gi 1 3 {} I \cdot II \cdot III\]
 Suppose $q \nmid RHS$. 

$\{q \nmid \gi 1 3 {} , q \nmid \g 1\} \implies q \mid \g 3$, by Lemma \ref{qfactor}. 
Now 
$q \nmid I$, $q\nmid \g 1 \g 2 \implies q \nmid \g{45}$, by Lemma \ref{qfactor}. 
\\
$q\nmid II$, but $q \mid \g 3 \implies q \mid \g 2$ or $q \mid  \g 4$  by Lemma \ref{qfactor},   $q \mid \g 4$. 

We already concluded $q \nmid \g {45}$, so together with $q \mid \g 4$, we have $q \nmid \g 5$. 

$q \nmid III, q \nmid \g 1 , \g 5 \implies  q \nmid \g {34}$, by Lemma \ref{qfactor}.\\
But $q \nmid \g {34}$ together with $q \mid \g 4$, implies  $q \nmid \g 3$, a contradiction.

{\bf Inequality (\ref{eq:5}).}
\[ 
 \gi 1 2 {} \leq \gi 1 3{}
(\gi 2 4 3 )
 (\gi 1 5 4 ) 
 (\gi 1 2 {(35)}  )
 (\gi 2 3 {(45)}) \]
 \[ =  \gi 1 3 {} I \cdot II \cdot III \cdot IV\]

Suppose $q \nmid RHS$. 

$\{q \nmid \gi 1 3 {} , q \nmid \g 1\} \implies q \mid \g 3$, by Lemma \ref{qfactor}. 

Lemma \ref{qfactor} on $III$ gives $q \nmid \g {35}$. 
\\
$q \mid \g 3, q \nmid \g {35} \implies q \nmid \g 5$.
\\
Lemma \ref{qfactor} on $II$ gives $q \nmid \g 4$, and on $I$ gives $q \nmid \g 3$, a contradiction.

{\bf Inequality (\ref{eq:6}).}
\[
\gi 1 2 {} \leq \gi 1 3 {}
 (\gi 2 4 5 )
 (\gi 4 5 3  ) 
 (\gi 1 2 {(34)}  )
 (\gi 1 3 {(45)})  \]
 \[=
\gi 1 3 {} I \cdot II \cdot III \cdot IV
\]
Suppose $q \nmid RHS$. 
$\{q \nmid \gi 1 3 {} , q \nmid \g 1\} \implies q \mid \g 3$, by Lemma \ref{qfactor}.  
Now $q \nmid III \implies q \nmid \g {34}$, by  Lemma \ref{qfactor}, then
$q \mid \g 3, q \nmid \g {34} \implies q \nmid \g 4$. \\
Applying Lemma \ref{qfactor} to factor $I$ then $II$, we get a contradiction. 

{\bf Inequality (\ref{eq:7}).}
\[
\gi 1 2 {} \leq \gi 2 4 {}
 (\gi 1 3 4 )
 (\gi 1 5 3  ) 
 (\gi 2 4 {(35)}  )
 (\gi 1 2 {(45)})  \]
 \[= \\
\gi 2 4 {}  I \cdot II \cdot III \cdot IV
\]
 Suppose $q \nmid RHS$. 

$\{q \nmid \gi 2 4 {} , q \nmid \g 2\} \implies q \mid \g 4$, by Lemma \ref{qfactor}. 

 $q \nmid IV \implies  q \nmid G_{45} \implies q \nmid G_5$.  \\
 $\{q \nmid II, q\nmid \g 5  \} \implies q \nmid \g 3$. \\
$\{q \nmid I, q \nmid \g 3 \}\implies q \nmid \g 4$ a contradiction.

{\bf Inequality (\ref{eq:8}).}
This and subsequent inequalities have a different form:
\[ 
{\gi 1 2 {}  \gi 1 2 {} }
 \leq 
{ (\gi 3 4 {} )
(\gi 5 {(34)} {} )
}
 (\gi 1 2 3  ) 
 (\gi 1 2 4  )
 (\gi 1 2 5)  = \]
 \[
( I \cdot II) \cdot III \cdot IV \cdot V
\]

The following claim is a generalization of claim \ref{qrhs} to inequalities of this type. 
\begin{claim} 
In order to prove the inequality, it suffices to show that the  power of $q$ dividing $I \cdot II \cdot III \cdot IV \cdot V$ is at least 2. 
\end{claim}
\begin{proof}
The LHS is  $(\gi 1 2 {}  )= (\frac q p)^2$, while each of the factors on the right hand side takes values among $\{1,p, q/p, q, pq\}$. Clearly if at least two of them are among $q/p, q, pq$, the inequality holds for each $I, II, III, IV, V \geq 1 $ by Lemma \ref{infvalues}. This proves the claim. 
\end{proof}

Next we show that condition of the claim is indeed satisfied. 

\begin{claim} The power of $q$ dividing $I \cdot II \cdot III \cdot IV \cdot V$ is at least 2. 
\end{claim}

Recall that $\gi a b {}$ takes values $\{1, p\}$ only when $q \mid \g a \g b$, by Lemma \ref{infvalues}.

Suppose $q \nmid I \gi 3 4 {} $. \\
$q \nmid I \implies q\mid \g 3$  or $q\mid \g 4$. \\
$q \mid \g 3 \implies q \mid III$, by Lemma \ref{qfactor}.\\
$q \mid \g 4 \implies q \mid IV$, by Lemma \ref{qfactor}. \\ 
Hence $q \mid I\cdot III\cdot IV$, so if $q \nmid I$, then it must divide either $III$ or $IV$. 
\\
Next suppose $q \nmid II $.

$q \nmid II \gi 5 {34} {} \implies q\mid \g 5$  or $q\mid \g {43}$. \\
$q \mid \g 5 \implies q \mid V$, by Lemma \ref{qfactor}.\\
$q \mid \g {34} \implies q \mid \g 4 \implies q \mid IV$, by Lemma \ref{qfactor}. \\ 

Hence $q \mid II \cdot V \cdot IV$, so if $q$ does not divide $II$, then it must divide either $V$ or $IV$. 

The problem now is now what if $IV$ is used to ``make up for" $q\nmid I$ and $q\nmid II$ simultaneously?

So suppose $q \nmid I, q\nmid II, q\mid IV$. We must demonstrate that $q$ divides at least one other factor. 
\[ 
{\gi 1 2 {}  \gi 1 2 {} }
 \leq 
{ (\gi 3 4 {} )
(\gi 5 {34} {} )
}
 (\gi 1 2 3  ) 
 (\gi 1 2 4  )
 (\gi 1 2 5)  = \]
 \[
( I \cdot II) \cdot III \cdot IV \cdot V
\]
We have $q\nmid III \implies q \nmid \g 3$, and
$q \mid IV \implies q \mid \g 4 $, so
together $q \nmid \g {34}$. \\
Then $q \nmid \g {34}$, $q \nmid II \implies q\mid \g 5 \implies q \mid V$
and we found another term, as desired.

{\bf Inequality (\ref{eq:9}).}
\[ 
{\gi 1 2 {}  \gi 1 2 {} }
 \leq 
{ (\gi  1 3 {} )
(\gi  4 5  {})
}
 (\gi 1 2 4  ) 
 (\gi 1 2 5  )
 (\gi 1 {(45)} 3)  = \]
 \[
( I \cdot II) \cdot III \cdot IV \cdot V
\]
Suppose $q \nmid I$, then $q \mid \g 3$.\\
$q \mid \g 3 \implies q \mid V$ or $q \mid \g {45}$. \\
$q \mid \g{45} \implies q \mid \g 5 \implies q \mid IV$.\\ 
$q \mid \g{45} \implies q \mid \g 4 \implies q \mid III$.\\ 
We conclude $q \nmid I \implies q \mid  V $ or ($ q \mid IV $ and $q \mid III$). \\
Suppose $q \nmid II$, then $q \mid \g 4 \g 5 \implies q \mid III\cdot IV$.
\\
$q \nmid I $  and $q \nmid  II$ results in at least two factors of $III, IV, V$ which $q$ divides, as desired.

{\bf Inequality (\ref{eq:10}).}
\[ 
{
(\gi 1 2 {})
(  \gi 1 2 {} )}
 \leq 
 \]
 \[
{ (\gi  3 4  {} )
(\gi  1 5  {})
}
 (\gi 1 2 3  ) 
 (\gi 1 2 4 )
 (\gi 2 4 5 ) 
 (\gi 1 3 {(45)}) = \]
 \[
( I \cdot II) \cdot III \cdot IV \cdot V \cdot VI
\]
 
Suppose $q \nmid I$, then $q \mid \g 3 \g 4$. \\
$q \mid \g 3 \g 4 \implies q \mid III \cdot IV$. \\
Suppose $q\nmid II$, then $q\mid \g 5$. 
\\
$q \mid \g 5 \implies $ either $q \mid V$ or $q \mid \g 4$. 
$q \mid \g 4 \implies q \mid IV$. \\
We summarize: $q \nmid II \implies  q \mid IV \cdot V$.

Suppose $q \nmid I$ and $q \nmid II$. We want to show that $IV$ cannot be the only factor divisible by $q$. 
Suppose then $q \nmid I$ and $q \nmid II$, $q \mid IV$. We show that $q$ divides another factor. 
For the sake of contradiction, suppose $q$ only divides $I, II, IV$. 
$q \nmid II \implies q \mid \g 5$. \\
$q \mid IV \implies q \mid \g 4$.\\
$q \nmid III \implies q \nmid \g 3$.\\
$q \nmid VI \implies q \nmid \g {45}$.\\
$q \nmid \g {(45)}  , q \mid \g 4 \implies q \nmid \g 5$,\\
a contradiction to $q \nmid II$. 

\end{proof}

\section{Groups of Order $p^2q$}
\label{sec:p2q}

Next we consider groups of order $p^2q$,  $p,q$ two distinct primes. 
We are interested in seeing how far the techniques developed for the case $pq$ can be extended for the case $p^2q$.
Unfortunately we cannot derive a complete generalization of Proposition \ref{prop:pq} and show that DFZ inequalities necessarily hold for these groups. However, we are able to use similar reasoning to derive partial results, which can be used to reduce computation when looking for violators of DFZ inequalities.
The case $|G_1|=|G_2|=p$ remains open.

\begin{lem}
\label{uniqueSylowq}
Let $p< q$ and let $G$ be a group of order ${\pp}^2{\qq}$ or $pq^2$. Then $G$ has at least one normal Sylow subgroup. 
\end{lem}

\begin{proof}
{\bf Case $|G| = {\pp}{\qq}^2$, ${\pp}<{\qq}$.} 
We show that the Sylow ${\qq}$-subgroup is normal. 
Indeed, the number $n_{\qq}=1+k{\qq}$, $k \geq 0$, of Sylow ${\qq}$-subgroups in a group of order ${\pp}{\qq}^2$ divides ${\pp}<{\qq}$, hence $n_{\qq}=1$. 
 
{\bf Case $|G| = {\pp}^2{\qq}$, ${\pp}<{\qq}$.} We have $n_{\qq}=1+k{\qq}$ divides ${\pp}^2$ and $n_{\pp}=1+l{\pp}$ divides ${\qq}$, $k,l \geq 0$, thus $n_{\qq}\in\{1,{\pp},{\pp}^2\}$. 
 \begin{itemize}
 \item 
 Case  $n_{\qq}=1$. Then the Sylow ${\qq}$ subgroup is unique, and we are done. 
 \item
 Case  $n_{\qq}={\pp}$ cannot occur for $1+k{\qq}={\pp}$ will contradict ${\pp}<{\qq}$.  
 \item
 Case $n_{\qq}={\pp}^2$. We  show that this forces $n_{\pp} = 1$, i.e. the Sylow ${\pp}$-subgroup to be unique.  
 Note that the number of elements of order ${\qq}$ is $({\qq}-1){\pp}^2$, thus the number of elements of order ${\pp}^t$, $t=0,1,2$, is ${\pp}^2{\qq}-({\qq}-1){\pp}^2={\pp}^2$. But that is precisely the order of any Sylow ${\pp}$-subgroup. Therefore all elements of order ${\pp}$ are included in a single Sylow ${\pp}$-subgroup. Hence there is a unique Sylow ${\pp}$-subgroup and we are done.
\end{itemize}
\end{proof}

We have concluded that $G$ has at least one unique Sylow subgroup - and it can correspond to either of the prime divisors of $G$. 
As per Remark \ref{rem:uniqueSylow}, we will try to take advantage of the uniqueness of this Sylow subgroup and derive a generalization of Proposition \ref{prop:pq}. Let us hence  reserve the prime $q$ for the cases when the Sylow $q$-subgroup is necessarily unique and express the order of $G$ accordingly as either $p^2q$ or $pq^2$: we write $|G| = p^2q$ when the Sylow subgroup of prime order is unique (while the Sylow subgroup of squared order may or may not be unique). Similarly, we will write $G = pq^2$ when the Sylow subgroup of squared order is unique (while the Sylow subgroup of prime order may or may not be). Note that a priori both cases $pq^2$ and $p^2q$ could occur with either $p>q$ or $q>p$.

We note immediately that if both the Sylow $q$-subgroup and the Sylow $p$-subgroup are unique, then $G$ is abelian, and by Corollary \ref{cor:ci} we know that the 10 DFZ inequalities hold. In what follows, we henceforth always suppose that the Sylow $p$-subgroup is not normal.

We are thus left to consider two cases: 
\begin{enumerate}
\item
The Sylow $q$-subgroup is normal and $|G|= p^2q$.
\item
The Sylow $q$-subgroup is normal and $|G|= pq^2$.
\end{enumerate}

\subsection{Case 1: the Sylow $q$-subgroup is normal and $|G|= p^2q$.}

\begin{claim}If $p>q$, then $G$ is in fact abelian. Hence the DFZ inequalities hold for $G$.
\label{ppq_pgq}
\end{claim} 
\begin{proof}
Since the Sylow $q$-subgroup is by assumption normal, it suffices to show that a Sylow $p$-subgroup is normal (and hence unique). But a Sylow $p$-subgroup $G_p$ is of order $p^2$, so it has index $[G:G_p]=q$, which is the smallest prime that divides the order of $G$. It follows that $G_p$ is normal in $G$. 
\end{proof}

We may therefore assume that $p<q$. 

\begin{lem}\label{lem:subgroupppq}
Let $G$ be a group of order $p^2q$ with unique Sylow $q$-subgroup. 

If $|G_{ac}|$ or $|G_{bc}|$ takes values in $\{1,p^2q,q,pq\}$, then $G_{ac}G_{bc}$ is a subgroup. 

The case when $G_{ac}G_{bc}$ is not a subgroup may then only happen when $(|G_{ac}|,|G_{bc}|)$ takes value in $\{(p,p),(p,p^2),(p^2,p^2)\}$.
\end{lem}
\begin{proof}
By Claim \ref{ppq_pgq} when $p>q$ the group $G$ is abelian, hence $G_{ac}G_{bc}$ is always a subgroup. 
We may therefore assume $p<q$.
In order to show that $G_{ac}G_{bc}$ forms a subgroup, it suffices to show that one of $G_{ac}$, $G_{bc}$ is normal. 
We will show that any subgroup $H$ of order $1,q,pq,p^2q$ is normal.

Orders $1,p^2q$ follow trivially. 

Suppose  $H$ has order $q$. Then it is the unique  Sylow $q$-subgroup, which is normal. Hence $G_{ac}G_{bc}$ forms a subgroup.

Suppose  $H$ has order $pq$. Then it is a subgroup of order $pq$, which has index $p$, the smallest prime that divides $|G|$. Thus it is normal.

The second claim of the lemma follows from removing $\{1,p^2q,q,pq\}$ from the list  $\{1,p,q,p^2,pq,p^2q\}$ of possible orders for $G_{ac}$ and $G_{bc}$ .
\end{proof}

\begin{prop}\label{prop:intersectSylow}
Let $G$ be a group of order $p^2q$ with a normal Sylow $q$-subgroup $G_q$. 
Then any two Sylow $p$-subgroups intersect in a subgroup of order $p$.
\end{prop}

\begin{proof}

It is enough to show that $G$ contains a normal subgroup of order $p$. Indeed, let $N$ be a normal subgroup of order $p$, it is contained in some Sylow $p$-group $G_p$. Let $G'_p$ be another Sylow $p$-subgroup, which is therefore a conjugate of $G_p$, but since the conjugate of $N$ is $N$ itself, it is contained in $G'_p$. 

We are left to show that $G$ contains a normal subgroup of order $p$ which is normal.

{
Since $G_q $ is normal in $G$ we have the exact sequence 
\[1 \rightarrow G_q \rightarrow G \rightarrow G_p\rightarrow 1\]
where we use $G/G_q \cong G_p$. By elementary group theory we know that $G_p$ acts on $G_q$ by conjugation. In other words, there is a homomorphism  
\[\alpha: G_p \rightarrow Aut(G_q)\] from $G_p$ to the automorphism group of $G_q$. \\
Case1 : If $G_p$ is cyclic of order $p^2$ then it has a unique subgroup of elements of order $p$. But that means that this subgroup is normal in $G$ because the conjugate of any element of order $p$ is another element of order $p$. \\
Case 2: $G_p = C_p \times C_p$.\\
Now $Aut(G_q) \cong C_{q-1}$ is cyclic. Hence $\alpha $ has a kernel $K \leq S_p$. But then $K$ acts trivially on $G_q$, in other word $K$ is normal in $G$. 
}
\end{proof}

Next we get an analogous result as that of Lemma \ref{infvalues}.

\begin{prop}\label{infvaluesppq} 
Let $G$ be a group of order $p^2q$ with a normal Sylow $q$-subgroup. 
Then $\gi abc \geq 1$ and
\[ 
{\gi a b c} = \frac {|G_c|} {|G_{ac}G_{bc}|}  \in  \{1,p,p^2,\frac{q}{p},q,pq,p^2q\}. 
\]
Furthermore ${\gi a b c}=\frac{q}{p}$ if and only if 
\begin{itemize}
\item
$|G_c|=pq$, $|G_{ac}|=|G_{bc}|=p$.
\item
$G_c=G$, $|G_{a}|=p$, $|G_{b}|=p^2$, and $G_{ab}=1$.
\end{itemize}
\end{prop}
\begin{proof}
The fact that $\gi a b c \geq 1 $ follows from the fact that $G_{ac}G_{bc}$ is a subset of $G_c$. 
Then
\[ 
{\gi a b c} = \frac {|G_c|} {|G_{ac}G_{bc}|}  \in  p^{\{-2, -1, 0, 1,2 \} }q^{\{0,1\} }. 
\]
Indeed, the restrictions on powers of $q$ follow from Proposition \ref{valq} since the Sylow $q$-subgroup is abelian, while the restrictions on powers of $p$ come from the inclusions $G_{ac} \leq G_c$, $G_{bc}\leq G_c$ and the fact that $v_p(G_c) \in \{ 0,1,2\}$.
The cases $\frac{1}{p}<1$ and $\frac{1}{p^2}<1$ cannot happen. 

By Claim \ref{ppq_pgq}, the case $p>q$ is abelian, thus the result follows. We may therefore assume $p<q$. 

We are left to exclude the case $\frac{q}{p^2}$, and to characterize when $\frac{q}{p}$ occurs. Since $p$ cannot possibly divide $q$, these ratios appear only when $G_{ac}G_{bc}$ is not a subgroup of $G_c$. By Lemma \ref{lem:subgroupppq}, this may happen when $(|G_{ac}|,|G_{bc}|)$ is $(p,p)$, $(p,p^2)$ or $(p^2,p^2)$.
\begin{itemize}
\item $(p,p)$ \\
$\gi a b c = \frac{1\cdot |G_c|}{p \cdot p }$. Since $G_{ac},G_{bc}$ are subgroups of $G_c$, we have $v_p(G_c)\geq 1$, and to obtain $q/p$ or $q/p^2$, we need $v_q(G_c)=1$. If $|G_c|=pq$, $\gi a b c=\frac{q}{p}$, if $|G_c|=p^2q$, $\gi a b c = q$.
\item $(p,p^2)$ \\
Since $G_{bc}$ is a subgroup of $G_c$, we have $v_p(G_c)=2$, and to obtain $q/p$ or $q/p^2$, we need $v_q(G_c)=1$. Thus $G_c=G$.\\
if $G_{abc} = 1$, $\gi a b c = \frac{1\cdot p^2q}{p \cdot p^2 }=q/p$.\\
if $G_{abc} = p$, $\gi a b c = \frac{p\cdot p^2q}{p \cdot p^2 } = q$.

\item $(p^2,p^2)$ \\ 
For the same reason as above, $G_c=G$. \\
if $G_{abc} = 1$, $\gi a b c = \frac{1\cdot p^2q}{p^2 \cdot p^2 } = q/p^2$, \\
if $G_{abc} = p$, $\gi a b c = \frac{p\cdot p^2q}{p^2 \cdot p^2 } = \frac q p $.
\end{itemize} 

The ratio $q/p$ thus exactly happens when the intersection is trivial, while the ratio $q/p^2$ exactly happens with the intersection of two Sylow $p$-subgroups is trivial, a case that never occurs, by Proposition \ref{prop:intersectSylow}.
\end{proof}

This is a partial generalization of Proposition \ref{prop:pq} on groups of order $pq$. Note that reducing to the case 
$|G_1|=|G_2|=p$ can be used to reduce computation when looking for violators of DFZ inequalities. 

\begin{prop}\label{prop:dfzppq}
Let $G$ be a group of order $p^2q$ with normal Sylow $q$-subgroup. 
If $|G_1|=|G_2|=p$ does not hold,
then the 10 DFZ inequalities hold for $G$. 
\end{prop}
\begin{proof} 
Recall that when either $G_1, G_2$ is normal or $G_1 = G_2$, then $G_1G_2$ is a subgroup and the inequalities hold. 
The case $p>q$ follows from Claim \ref{ppq_pgq}.
We now assume $p<q$. 

Recall that the DFZ inequalities hold when one of the groups $G_1, G_2$ is normal. 
By Lemma \ref{lem:subgroupppq}, we narrow down the non-normal candidates for $G_1, G_2$ to subgroups of orders $(p,p), (p, p^2), (p^2,p^2)$. 

Also by Proposition \ref{infvaluesppq} 
\[ {\gi a b c} = \frac {|G_c|} {|G_{ac}G_{bc}|}  \in  \{1,p,p^2,\frac{q}{p},q,pq,p^2q\}. 
\]

Evaluating $\gi 1 2 {} $ in the LHS of the 10 DFZ inequalities (\ref{eq:1})-(\ref{eq:10}) we show that we only need to consider the case when the LHS is $\frac q p$:
\begin{itemize}
\item $(p,p^2)$ \\
\begin{itemize}
\item $G_{12} = 1$: $\gi 1 2 {} = \frac{1\cdot (p^2 q)}{p \cdot p^2 } = \frac q p $ \\
\item $G_{12} = p$:  we are done
since $G_1$ has order $p$, and it intersects $G_2$ in a subgroup of order $p$, implying that $G_1$ is a subset of $G_2$, and inequalities hold by Corollary \ref{cor:ci}. 
\end{itemize}

\item $(p^2,p^2)$ 
\begin{itemize}
\item
$G_{12} = 1$ cannot happen since no  
{two Sylow $p$-subgroups intersect trivially by Proposition \ref{prop:intersectSylow}.}\\  
\item $G_{12} = p$:  
$\gi 1 2 {} = \frac{p\cdot (p^2 q)}{p^2 \cdot p^2 } = \frac q p $
\end{itemize} 
\end{itemize} 

{
\begin{claim}\label{sufficesvq}
Suppose that the LHS is $\frac{q}{p}$. If $v_q(RHS) \geq 1$, then the inequalities hold. 
\end{claim}
Terms on the right hand side 
can only take values among $\{1,p,q,pq,p^2,p^2q,\frac{q}{p}\}$ by Proposition \ref{infvaluesppq}.
The assumption $v_q(RHS) \geq 1$ means that at least one term on the RHS is taking values among $\{q,pq,pq^2,\frac{q}{p}\}$. Since each of these is $\geq \frac q p$, while all other terms on the RHS are $\geq 1$, we conclude LHS $\leq $ RHS and we have proved the claim. $\square$
}

To complete the proof we must show that indeed $v_q(RHS) \geq 1$. This has  been shown for cases of groups of order $pq$ and the proofs are identical. 
\end{proof}

Let us comment on the case $(|G_1|,|G_2|)=(p,p)$.
Then $\gi 1 2 {} = \frac{1\cdot (p^2 q)}{p \cdot p } = q$ 
and the proof technique used above relying on $v_q(RHS) \geq 1$ is not sufficient anymore.
Say $G_1, G_2$ have orders $(p,p)$ so the LHS of an inequality is $q$. 
But on the RHS we may have $q/p$.
By Proposition \ref{infvaluesppq}, this will exactly happen 
when  
\begin{itemize}
\item
$|G_c|=pq$, $|G_{ac}|=|G_{bc}|=p$.
\item
$G_c=G$, $|G_{a}|=p$, $|G_{b}|=p^2$, and $G_{ab}=1$.
\end{itemize}
An inequality of the form
$q \leq q/{p} \cdot I \cdot II \cdot III$ where $q/p$ corresponds e.g. to $\gi 1 5 {}$ does not have to be true. It depends on the terms $I,II,III$, therefore the techniques developed in this paper do not apply immediately to this case. 

\begin{exmp}\rm
As a concrete example, consider the dihedral group $D_{20}=\langle r,s,~r^{10}=s^2=1\rangle$. There is one Sylow 5-subgroup, namely $\langle r^2 \rangle$. 
There are five Sylow 2-subgroups, given by $\langle s,r^5 \rangle=\{1,s,r^5,sr^5\}$ and its four conjugates, for example $r\langle s,r^5 \rangle r^{-1}=\{1,rsr^{-1},r^5,rsr^4\}$. We notice the subgroup $\{1,r^5\}$ which is in the intersection of both these Sylow 2-subgroups.
Take
\[
G_1=\{ 1, rsr^{-1}\},~G_2=\{1,s\},~G_5=\{1,s,r^5,sr^5\}.
\]
Then
\[
\gi 1 2 {}=\frac{20}{4}=5,~\gi 1 5 {}=\frac{20}{8}=\frac{5}{2}.
\]
\end{exmp}

\subsection{Case 2: the Sylow $q$-subgroup is normal and $|G|= pq^2$.}
Let us look at the other case, where the group $G$ has a normal Sylow $q$-subgroup, while the Sylow $p$-subgroup is not normal, and $G$ has order $pq^2$ (the normal Sylow subgroup corresponds to the squared prime). 

An analogous result as that of Lemma \ref{infvaluesppq} is obtained similarly.

\begin{lem}\label{infvaluespqq} 
Let $G$ be a group of order $pq^2$ whose Sylow $q$-subgroup is normal. 
Then $\gi abc \geq 1$ and
\[ {\gi a b c} = \frac {|G_c|} {|G_{ac}G_{bc}|}  \in \{1,p,q,pq,q^2,pq^2,\frac{q}{p},\frac{q^2}{p}\}.  
\]
\end{lem}
\begin{proof}
That $\gi a b c \geq 1 $ follows from the fact that $G_{ac}G_{bc}$ is a subset of $G_c$, as before. 
Then \[ {\gi a b c} = \frac {|G_c|} {|G_{ac}G_{bc}|}  \in  p^{\{-1, 0, 1\} }q^{\{0,1,2\} }. \]
The arguments is as before, the restrictions on powers of $q$ follow from Proposition \ref{valq} since the Sylow $q$-subgroup is abelian, while the restrictions on powers of $p$ come from the inclusions $G_{ac} \leq G_c$, $G_{bc}\leq G_c$ and the fact that $v_p(G_c) = 0,1$. 
The case $1/p <1$ cannot happen.
\end{proof}

Next we look at the  case when $|G| = pq^2$, $G_q$ is normal, while $G_p$ is not. If $q<p$, then, as shown below, the group is forced to be $A_4$. 

{\begin{claim} \label{qlpA4}
Suppose $|G| = pq^2$, $G_q$ is normal, while $G_p$ is not, and $q<p$. Then $G \cong A_4$. 
Hence 10 DFZ inequalities hold for $G$ \end{claim}
Since $G_q$ is normal and $G_p \cong C_p$, we have an exact sequence 
\[1 \rightarrow G_q \rightarrow G \rightarrow C_p \rightarrow 1 \]
and $C_p$ acts on $G_q$, i.e. there is a homomorphism 
\[\alpha: C_p \rightarrow Aut(G_q)\]
There are two cases to consider: \\
If $G_q = C_{q^2}$, it is easy to show that $q < p$ implies that $\alpha$ is trivial. This makes $G$ an abelian group, contradicting the assumption that $G_p$ is not normal. \\
If $G_q = C_q \times C_q$, we know that $Aut(G_q)$ is the general linear group $GL_2(\mathbb F_q)$ of order 
$(q^2-1)(q^2-q)$. For $p$ to divide that, together with the assumption $q<p$, we would need that $q=2$ and $p=3$. We arrive that the resulting group must be $A_4$ and we have shown the claim. 
It can be checked numerically that $A_4$ does not violate any of the 10 DFZ inequalities. \\
}

Therefore we may assume that $p<q$. Here is the analogue  of Lemma \ref{lem:subgroupppq}.

\begin{lem}\label{lem:subgrouppqq}
Let $G$ be a group of order $pq^2$ with normal Sylow $q$-subgroup. Assume that $p<q$.

If $|G_{ac}|$ or $|G_{bc}|$ takes values in $\{1,pq^2,q,q^2\}$, then $G_{ac}G_{bc}$ is a subgroup. The case when the set $G_{ac}G_{bc}$ is not a subgroup happens only when the orders $(|G_{ac}|,|G_{bc}|)$ take values in $\{(p,p),(p,pq),(pq,pq)\}$.
\end{lem}
\begin{proof}
It suffices to show that any subgroup $H$ of order  $1,pq^2,q,q^2$ is normal. 

Orders $1, pq^2$ follow trivially. 
For order $q$, it follows from the uniqueness of the Sylow $q$-subgroup. 
For $q^2$, this is because a subgroup of order $q^2$ has index $p$, the smallest prime that divides $|G|$, thus it is normal.
The second claim follows from removing $\{1,pq^2,q,q^2\}$ from the list  $\{1,p,q,q^2,pq,pq^2\}$ of possible orders.
\end{proof}

As a result, the only non-normal candidates for $G_1, G_2$ will give as possible values for $(|G_1|,|G_2|)$ the pairs $(p,p)$, $(p,pq)$, $(pq,pq)$. 
Examining the possibilities for the LHS of equations (\ref{eq:1})-(\ref{eq:10}) gives  
\begin{itemize}
\item $(p,p)$ \\
$\gi 1 2 {} = \frac{1\cdot (pq^2)}{p^2 } = q^2/p$ \\
\item $(p,pq)$ \\
if $G_{12} = 1$, $\gi 1 2 {} = \frac{1\cdot (pq^2)}{p \cdot pq } = q/p$ \\
if $G_{12} = p$, since $G_1$ has order $p$, and it intersects $G_2$ in a subgroup of order $p$, $G_1$ is a subset of $G_2$, and inequalities hold by Corollary \ref{cor:ci}. 
\item $(pq,pq)$ \\
if $G_{12} = 1$ then 
$\gi 1 2 {} = \frac{1\cdot (pq^2 )}{pq \cdot pq } = \frac{1}{p} $ which is not possible\\
if $G_{12} = p$ then 
$\gi 1 2 {} = \frac{p\cdot (pq^2 )}{pq \cdot pq } = 1 $.\\
if $G_{12} = q$ then 
$\gi 1 2 {} = \frac{q\cdot (pq^2 )}{pq \cdot pq } = q/p $.\\
\end{itemize}

Therefore, we may conclude  as in Proposition \ref{prop:dfzppq}.

\begin{prop} \label{prop:dfzpqq}
Let $G$ be a group of order $pq^2$ with normal Sylow $q$-subgroup. 
If $|G_1|=|G_2|=p$ does not hold,
then the 10 DFZ inequalities hold for $G$. 
\end{prop}

\begin{proof}
The case $q<p$ is taken care of by Claim \ref{qlpA4}. 
The proof for the case $p<q$, where we now apply Lemma \ref{infvaluespqq} instead of Lemma \ref{infvaluesppq},   is similar as in Proposition \ref{prop:dfzppq}.

\end{proof}

%
%
\section{Group Violations}
\label{sec:viol}

\subsection{Violations by $S_4$} We prove that the symmetric group $S_4$ violates two of the ten DFZ inequalities, and so does the projective group $PGL_2(\FF_q)$, for $q$ a prime power.

\begin{prop}
The symmetric group $S_4$ of permutations on $4$ elements violates (\ref{eq:1})
and (\ref{eq:3}).
\end{prop}
\begin{proof}
Consider $G=S_4$.
Rewrite the inequality (\ref{eq:1}) in its group form:
\begin{equation}\label{eq:1grp}
\Gc I(X_1;X_2) \leq \Gc I(X_1;X_2|X_3)\Gc I(X_1;X_2|X_4)\Gc I(X_3;X_4|X_5) \Gc I(X_1;X_5)
\end{equation}
which simplifies to
\[
|G_{12}||G_{13}||G_{23}||G_{14}||G_{24}||G_{35}||G_{45}|
\leq
|G_2||G_3||G_{123}||G_{15}||G_{124}||G_{4}||G_{345}|.
\]
Take
\[
\begin{array}{lcllcl}
G_1 &=& \langle  (3,4), (2,4,3) \rangle,& G_3 &=& \langle (1,2)(3,4), (3,4) \rangle \\
G_2 &=& \langle (1,3), (1,3,2) \rangle,& G_4 &=& \langle (1,3)(2,4), (2,4) \rangle \\
G_5 &=& \langle (1,4)(2,3), (1,3)(2,4) \rangle
\end{array}
\]
with $|G_1|=6$, $|G_2|=6$, $|G_3|=4$, $|G_4|=4$, $|G_5|=4$.
Then
\[
\begin{array}{ccllcl}
G_1\cap G_2 &=& \langle (2,3) \rangle, & G_1\cap G_3 &=& \langle (3,4) \rangle  \\
G_2\cap G_3 &=& \langle (1,2) \rangle, & G_1\cap G_4 &=& \langle (2,4) \rangle \\
G_2\cap G_4 &=& \langle (1,3) \rangle, &  G_1\cap G_5 &=& \{1\}
\end{array}
\]
\[
\begin{array}{ccllcl}
G_4\cap G_5 &=& \langle (1,3)(2,4) \rangle, & G_3\cap G_4 \cap G_5 &=& \{1\} \\
G_3\cap G_5 &=& \langle (1,2)(3,4) \rangle, & G_1\cap G_2 \cap G_3 &=& \{1\}.
\end{array}
\]
But for the trivial group, all the other intersection subgroups are of order 2.
The left hand side of (\ref{eq:1grp}) yields
\[
|G_{12}||G_{13}||G_{23}||G_{14}||G_{24}||G_{35}||G_{45}|
= 2 \cdot  2 \cdot  2 \cdot 2 \cdot 2 \cdot 2 \cdot 2 = 128
\]
while the right hand side is
\[
|G_2||G_3||G_{123}||G_{15}||G_{124}||G_{4}||G_{345}|=
6 \cdot 4 \cdot 1 \cdot 1 \cdot 1 \cdot 4 \cdot 1 = 96.
\]

Next, rewrite the inequality (\ref{eq:3}) in its group form
\[
\Gc I(X_1;X_2) \leq  \Gc I(X_1;X_3) \Gc I(X_1;X_2|X_4) \Gc I(X_2;X_5|X_3) \Gc I(X_1;X_4|X_3,X_5)
\]
which becomes 
\begin{equation}\label{eq:3grp}
|G_{12}||G_{14}||G_{24}||G_{23}||G_{135}||G_{345}|
\leq
|G_{2}||G_{13}||G_{124}||G_{4}||G_{253}||G_{1345}|.
\end{equation}
Take
\begin{itemize}
\item
$G_1=\langle (3,4), (2,4,3) \rangle$, $|G_1|=6$,
\item
$G_2= \langle (1,2)(3,4), (3,4) \rangle$, $|G_2|=4$,
\item
$G_3= \langle (1,2)(3,4), (1,4)(2,3), (1,3) \rangle$, $|G_3|=8$,
\item
$G_4= \langle (1,3), (1,3,2) \rangle$, $|G_4|=6$
\item
$G_5=\langle (1,3)(2,4), (2,4) \rangle$, $|G_5|=4$.
\end{itemize}
so that
\[
G_{12}=\langle (3,4) \rangle,~G_{14}=\langle (2,3) \rangle,~G_{13}=\langle (2,4) \rangle,~
G_{24}=\langle (1,2) \rangle,~G_{23}=\langle (1,2)(3,4) \rangle,
\]
\[
G_{135}=\langle (2,4) \rangle,~G_{345}=\langle (1,3) \rangle
\]
and the left hand side of (\ref{eq:3grp})
\[
|G_{12}||G_{14}||G_{24}||G_{23}||G_{135}||G_{345}|=
2\cdot  2\cdot  2\cdot   2\cdot 2 \cdot 2 =64
\]
which is strictly larger than
\[
|G_{2}||G_{13}||G_{124}||G_{4}||G_{235}||G_{1345}|=
4\cdot 2\cdot   1\cdot   6\cdot 1 \cdot 1 = 48.
\]

\end{proof}

Let $GL_2(\FF_p)$ be the general linear group of order 2 over $\FF_p$. 
It has order $(p^2-1)(p^2-p)$, and naturally acts on the projective line $\PP(\FF_p)$, whose $p+1$ elements can be written (in homogeneous coordinates) as $\{(0,1),(1,1),\ldots,(p-1,1),(1,0) \}$. This action descends to a faithful action of $PGL_2(\FF_p)$, the projective general linear group over $\FF_p$ of degree $2$, on $\PP(\FF_p)$, which permutes the elements of $\PP(\FF_p)$. This thus yields an injective homomorphism from $PGL_2(\FF_p)$ to the symmetric group $S_{p+1}$.
For $p=3$, $|PGL_2(\FF_3)|=24=|S_4|$, showing that $PGL_2(\FF_p) \simeq S_4$.

This suggests to look at the group $PGL_2(\FF_q)$, $q$ a prime power, to find examples of violators. Indeed, it is exactly known when $S_4$ is contained in $PGL_2(K)$, 
for $K$ an arbitrary field.

\begin{prop}\cite[2.5]{Serre}
Suppose that the characteristic of $K$ is prime to the order of $S_4$. 
Then $PGL_2(K)$ contains $S_4$ if and only if  $-1$ is a sum of two squares in $K$.
\end{prop}

It follows that whenever $p \geq 5$, then $PGL_2(\FF_q)$ contains $S_4$, for $q=p^r$, for $r \geq 1$. The proof of the following result is attributed to Henry Mann (the result holds for characteristic 2, with a different proof).

\begin{prop}
Suppose that the characteristic of $\FF_q$ is not 2.
Every element of $\FF_q$ is a sum of two squares.
\end{prop}
\begin{proof}
Consider the map $\phi$, which maps $x$ to $x^2$.
Suppose $\phi(a)=\phi(b)$, then $a^2=b^2 \iff (a-b)(a+b)=0 \iff a=b$ or $a=-b$. 
Therefore $\phi$ restricted to $\FF_q^*$ sends two distinct elements of $\FF_q^*$ to a single element, from which we deduce that exactly half of the elements of $\FF_q^*$ are squares. Since $0$ is also a square, the set $S$ of squares in $\FF_q$ contains $(q+1)/2$ elements.

Next, we look at the additive structure of $\FF_q$, and pick $a \in \FF_q$. Notice that $|S|=|a-S|=|\{a-s,~s\in S\}|$. Since $|S|+|a-S|>|\FF_q|$, $|S| \cap |a-S|$ is not empty. Take $s'=a-s$ in the intersection, then $a=s'+s$, and is therefore a sum of two squares.  
\end{proof}

Therefore another example of group violating this inequality is $PGL_2(\FF_5)$ of order 120, which turns out to be isomorphic to $S_5$ (the same proof as above holds).

\subsection{Smallest Violators.} From Corollary \ref{cor:ci}, abelian groups never violate the 10 DFZ inequalities. From Proposition \ref{prop:pq}, neither do groups of order $pq$ for $p,q$ two distinct primes. Thus until order 23 (included), only 8, 12, 16, 18, and 20 are orders where potential violators could exist. Groups of order 8 are known to be abelian group representable \cite{MTO} and can be ruled out as well.
Groups of orders 12,16,18,20 are left, apart 16, all of them fit in the category of groups of order $p^2q$. However, since the case of $|G_1|=|G_2|=p$ is still incomplete, we checked numerically that no violation of the 10 DFZ inequalities is to be found. Therefore the smallest violator is of order 24.

\subsection{Simultaneous Violators.}
One may wonder whether it is possible to find a simultaneous violator for the ten DFZ inequalities. We provide a negative answer.
\begin{defn}
If a finite group $G$ and subgroups $G_1,\ldots,G_n$ violate two or more inequalities, then $(G,G_1,\ldots,G_n)$ is called a \textit{simultaneous violator} for those inequalities.
\end{defn}
\begin{prop}
\label{pro:simvio}
There do not exist any simultaneous violators for (\ref{eq:1grp}) and(\ref{eq:3grp}).
\end{prop}

\begin{proof}
Suppose $(G,G_1,G_2,G_3,G_4,G_5)$ are a simultaneous violator of  (\ref{eq:1grp})
and (\ref{eq:3grp}). That is,
\[
|G_{12}||G_{13}||G_{23}||G_{14}||G_{24}||G_{35}||G_{45}|
>|G_2||G_3||G_{123}||G_{15}||G_{124}||G_{4}||G_{345}| \]
and
\[
|G_{12}||G_{14}||G_{24}||G_{23}||G_{135}||G_{345}| >
|G_{2}||G_{13}||G_{124}||G_{4}||G_{235}||G_{1345}|.\]

Since all these quantities are positive
the product of the inequalities yields:
\begin{eqnarray}
\label{eqn:pdt}
& &|G_{12}|^2|G_{23}|^2|G_{14}|^2|G_{24}|^2|G_{35}||G_{45}||G_{135}|  \nonumber \\
& > &|G_2|^2|G_3||G_{123}||G_{15}||G_{4}|^2|G_{124}|^2|G_{235}||G_{1345}|.
\end{eqnarray}
But using Lemma \ref{lem:lemma234} repeatedly (indicating its use by parentheses), we have
\begin{eqnarray*}
& &(|G_{12}|^2|G_{23}|^2)|G_{14}|^2|G_{24}|^2|G_{35}||G_{45}||G_{135}|\\
&\le & (|G_2||G_{123}||G_{12}||G_{23}|)|G_{14}|^2|G_{24}|^2|G_{35}||G_{45}||G_{135}|\\
& = &  |G_2||G_{123}||G_{23}|(|G_{12}||G_{24}|^2)|G_{14}|^2|G_{35}||G_{45}||G_{135}|\\
&\le & |G_2||G_{123}||G_2||G_{124}||G_{14}|^2|G_{24}|(|G_{23}||G_{35}|)|G_{45}||G_{135}|\\
&\le & |G_2|^2|G_{123}||G_{124}|(|G_{14}|^2|G_{24}|)(|G_{3}||G_{235}|)|G_{45}||G_{135}|\\
&\le & |G_2|^2|G_{123}||G_3||G_{235}||G_4||G_{124}|^2(|G_{14}||G_{45}|)|G_{135}|\\
& \le &|G_2|^2|G_{123}||G_3||G_{235}||G_4|^2|G_{124}|^2|G_{145}||G_{135}|\\
& \le & |G_2|^2|G_3||G_{123}||G_{15}||G_{4}|^2|G_{124}|^2|G_{235}||G_{1345}|,
\end{eqnarray*}
a contradiction to inequality (\ref{eqn:pdt}) which concludes the proof.
\end{proof}

\begin{cor}
There do not exist any simultaneous violators for all DFZ inequalities.
\end{cor}
\begin{proof}
If there exists a simultaneous violator $(G,G_1,G_2,G_3,G_4,G_5)$ for all DFZ inequalities, it violates  (\ref{eq:1grp})
and (\ref{eq:3grp}) simultaneously, a contradiction to Proposition \ref{pro:simvio}.
\end{proof}

%
%
%

\end{document}